\documentclass[a4paper,12pt]{article}
\pdfoutput=1
\usepackage{setspace}
\usepackage[utf8]{inputenc}
\usepackage{amsmath,amsmath,amsthm,enumerate,graphics,psfrag}

\usepackage[normalem]{ulem} 
\usepackage{cancel} 

\usepackage{color}
\definecolor{Aqua}{rgb}{0,1,1}
\definecolor{Black}{rgb}{0,0,0}
\definecolor{Blue}{rgb}{.0,.0,1}
\definecolor{Brown}{rgb}{.5,0,0}
\definecolor{DarkBlue}{rgb}{0,0.08,0.45} 
\definecolor{DarkGray}{gray}{.25}
\definecolor{DarkGreen}{rgb}{.17,.5,.05}
\definecolor{Fuchsia}{rgb}{1,0,1}
\definecolor{Gray}{gray}{.5}
\definecolor{Green}{rgb}{0,.5,0}
\definecolor{LightCyan}   {rgb}{0.88,1.,1.}
\definecolor{LightGray}{gray}{.75}
\definecolor{Lime}{rgb}{0,.5,.5}
\definecolor{Magenta}{rgb}{1,0,1}
\definecolor{MyDarkGreen}{rgb}{0,0.4,0}
\definecolor{Navy}{rgb}{0,0,.5}
\definecolor{Olive}{rgb}{.5,.5,0}
\definecolor{Orange}      {rgb}{1.,0.65,0.}
\definecolor{PaleGreen}   {rgb}{0.6,0.98,0.6}
\definecolor{Pink}        {rgb}{1.,0.75,0.8}
\definecolor{Purple}{rgb}{.5,0,.5}
\definecolor{Red}{rgb}{1,0,0}
\definecolor{Silver}{gray}{.75}
\definecolor{Teal}{rgb}{0,.5,.5}
\definecolor{VioletRed}{rgb}{.208,0.125,.144}
\definecolor{YellowOrange}{rgb}{.82,.5,.11}

\usepackage{graphicx}
\usepackage{caption}
\usepackage{comment}

\newcommand{\poa}{{\rm PoA}}

\newcommand{\spoa}{{\rm SPoA}}

\newcommand{\opt}{{\rm OPT}}

\newcommand{\bea}{\begin{eqnarray*}}
\newcommand{\eea}{\end{eqnarray*}}
\newcommand{\beq}{\begin{equation}}
\newcommand{\eeq}{\end{equation}}
\newcommand{\cala}{{{{\cal A}}}}
\newcommand{\calb}{{{{\cal B}}}}
\newcommand{\calc}{{{{\cal C}}}}

\newcommand{\calg}{{{{\cal G}}}}

\newcommand{\calm}{{{{\cal M}}}}
\newcommand{\caln}{{{{\cal N}}}}
\newcommand{\calsn}{{\cal{SN}}}
\newcommand{\calo}{{{{\cal O}}}}
\newcommand{\calp}{{{{\cal P}}}}

\newcommand{\calr}{{{{\cal R}}}}
\newcommand{\cals}{{{{\cal S}}}}

\newcommand{\area}{{\it a}}

\newcommand{\nfdh}{\mbox{{\rm NFDH}}}

\newcommand{\eps}{\varepsilon}

\newcommand{\soc}{{\rm sc}}

         {\begin{list}%
                {}%
                {%
                 \setlength{\partopsep}{1mm}%
                 \setlength{\topsep}{0mm}%
                 \setlength{\parsep}{0mm}%
                 \setlength{\leftmargin}{10mm}%
                }%
          \item}%
         {\end{list}}

\newtheorem{theorem}{Theorem}[section]

\newtheorem{lemma}[theorem]{Lemma}
\newtheorem{corollary}[theorem]{Corollary}

\newcommand{\valor}{x}

\date{}
\begin{document}
 
\title{\sc Prices of anarchy of \\ selfish 2D bin packing
    games\protect\footnote{\ This research was partially supported by
      CNPq (Proc.~456792/2014-7, 311499/2014-7, 305705/2014-8, 308116/2016-0, 306464/2016-0, 425340/2016-3),
FAPESP (Proc.~2013/03447-6, 2015/11937-9, 2016/01860-1, 2016/23552-7) and Project MaCLinC of NUMEC/USP.}}
  \author{
Cristina G. Fernandes\thanks{\ Institute of Mathematics and Statistics, University of São Paulo, Brazil, \texttt{\{cris,cef,yw\}@ime.usp.br}.}  \  \quad
Carlos E. Ferreira$^{\dag}$  \\
Flávio K. Miyazawa\thanks{\ Institute of Computing, University of Campinas, Brazil, \texttt{fkm@ic.unicamp.br}.} \ \quad
Yoshiko Wakabayashi$^{\dag}$}
 
\maketitle

\begin{abstract}
  We consider a game-theoretical problem called selfish 2-dimensional
  bin packing game, a generalization of the 1-dimensional case
  already treated in the literature. In this game, the items to be
  packed are rectangles, and the bins are unit squares. The game
  starts with a set of items arbitrarily packed in bins. The cost of
  an item is defined as the ratio between its area and the total
  occupied area of the respective bin. Each item is a selfish player
  that wants to minimize its cost. A migration of an item to another
  bin is allowed only when its cost is decreased.
  We show that this game always converges to a Nash equilibrium (a
  stable packing where no single item can decrease its cost by
  migrating to another bin).  We show that the pure price of anarchy
  of this game is unbounded, so we address the particular case where
  all items are squares.  We show that the pure price of anarchy of
  the selfish square packing game is at least~$2.3634$ and at
  most~$2.6875$. We also present analogous results for the strong Nash
  equilibrium (a stable packing where no nonempty set of items can
  simultaneously migrate to another common bin and decrease the cost
  of each item in the set).  We show that the strong price of anarchy
  when all items are squares is at least~$2.0747$ and at
  most~$2.3605$.

\smallskip

\end{abstract}


\section{Introduction}

The advent of the Internet and its increasing use have brought new
computational tasks and different ways of performing various activities.  
In such a decentralized environment, many activities involve competition and
collaboration over the resources, so users may act selfishly to maximize their
benefits.  This behavior suggests game-theoretical models for a number of applications.

Many new ideas, approaches, and models of analysis were first proposed to bin packing problems, 
which motivate the study of game theoretical versions of these problems.

We investigate here a class of packing games which we call~\emph{selfish 2D bin packing games}, 
in the special case where the bins (or recipients) are unit squares and 
the items to be packed are rectangles, or more particularly, squares.  
We call the corresponding games~\emph{selfish rectangle packing} or \emph{selfish square packing}. 
(The more general 2D class may include cases in which the items are $2$-dimensional objects 
of other specific forms, such as disks or triangles.)  
All games start with a set of items packed in bins.  
In our model, the cost of an item is defined as the ratio between its area
and the total occupied area of the respective bin.  
Each item is a selfish player: it wants to minimize its cost and, 
for that, it may selfishly migrate to a bin with a better occupied area. 
A stable packing in which no item can decrease its cost by migrating
to another bin is called a \emph{Nash equilibrium}.

Koutsoupias and Papadimitriou~\cite{KoutsoupiasP99,KoutsoupiasP09}
were the first to study a measure in a game-theoretic framework that
nowadays is known as the \textit{price of anarchy}\footnote{\
  In~\cite{KoutsoupiasP99}, the authors used the term
  \emph{coordination ratio} for the concept, which was later called
  \emph{price of anarchy} by Papadimitriou~\cite{Papadimitriou01}.},
which is the ratio between the worst social cost (number of used bins)
of a Nash equilibrium and the optimal social cost (minimum number of
bins needed to pack all items).  When the game admits coalitions,
these concepts are known as the \textit{strong Nash equilibrium}
and \textit{strong price of anarchy}.
In non-cooperative games, the price
of anarchy measures the loss of the overall performance due to the
decentralized environment and the selfish behavior of the players.

The 1-dimensional (1D) version of the game described above is known as
the \emph{selfish (1D) bin packing game}. It was first investigated by
Bil\`o~\cite{Bilo06}, who proved that this version of the game admits
a pure Nash equilibrium. He showed upper bounds on the number of steps
to reach a Nash equilibrium from an arbitrary initial configuration
and also showed that the price of anarchy is at least $1.6$ and at
most $1.666$. The lower bound and the upper bound have been improved
by Yu and Zhang~\cite{YuZ08} and by Epstein and
Kleiman~\cite{EpsteinK11} to $1.6416$ and $1.6428$, respectively.
They also showed that the strong price of anarchy is at least $1.6067$
and at most $1.6210$.
There are also results in the literature, with different cost functions.
Ma et al.~\cite{MaDHTYZ13} 
studied the model in which all items in the same bin share the cost equally, 
that is, if a bin contains $k$ items, then the cost for each item in this bin is $1/k$.  
They showed that the price of anarchy of any Nash equilibrium under this cost function has an upper bound of $1.7$ and also that it is possible to obtain a Nash equilibrium from a feasible packing in $O(n^2)$ steps without increasing its social cost. This result leads to an algorithm that obtain a Nash equilibrium in time $O(n^2)$ with price of anarchy that is at most  $1+\epsilon$, for any given $\epsilon>0$.
For a survey on selfish packing games, we refer the reader 
to Epstein~\cite{Epstein13}.



          
In 1D packing, it is easy to decide if an item fits in a unit bin with
other items, but this is not the case for higher dimensional packing.
So we also consider a parameterized version of the original game in which 
the players use a specific (polynomial-time) packing algorithm to decide on moving an item. 
One can think of the original game as one parameterized by an exact packing algorithm. 
As we will see, some results obtained with the use of specific 
packing algorithms yield results for the original generic game.

For the rectangle packing game, we consider oriented packing (that is,
when rotations are not allowed).  We show that, in this case, the game
converges to a Nash equilibrium. Then, we prove that its price of
anarchy is unbounded. In view of this, we consider two particular
cases: (a) the parametric version in which the dimensions of the
rectangles are bounded by~$1/m$, for an integer $m\geq 2$; and (b) the
case in which all items are squares.  For the first case, we
show that the game parameterized by the well-known \nfdh\ (Next Fit
Decreasing Height) algorithm~\cite{CoffmanGJT80} (and the original
game as well) has price of anarchy at most $(m/(m-1))^2$.  For the
latter case, the selfish square packing game, we prove that the pure
price of anarchy is at least $2.3634$ and at most~$2.6875$.

We also present results on the strong Nash equilibrium.  We prove that
the strong price of anarchy of the selfish square packing game is at
least~$2.0747$ and at most~$2.3605$.

We have presented some preliminary results on the selfish 2D packing
game in an extended abstract for a conference~\cite{FFMW11}; this
paper contains improved results and new ones.

\section{Problem definition and preliminaries}\label{description}

The games considered here are on packing of rectangles (or squares)
into unit squares.  They are natural generalizations of the selfish 1D
packing game, that was introduced by Bil\`o~\cite{Bilo06}.
We observe that, as we consider oriented packing, there is no loss of
generality in the assumption that the bins are unit squares (if not, a
scaling of the items and bins will reduce the instance to this case).

We denote by $\{1,2,\ldots,n\}$ the set of items (rectangles) in these
games, and assume that each item~$i$ has area $a_i>0$ and fits in a
bin. Each item is a selfish player that wants to minimize its cost
(defined in what follows).

A \emph{configuration} of a game in a certain moment is a sequence
$p=(p_1,\ldots,p_n)$, where $p_i$ indicates the bin selected by~$i$. For a bin
$k$, we denote by~$R(k)$ the set of rectangles~$i$ for which $p_i=k$. If the
rectangles in $R(k)$ fit all together in a bin, then $\area(R(k))$ denotes
the sum of the areas of the rectangles in~$R(k)$, otherwise, we define that
$\area(R(k))=0$, and refer to $k$ as an \emph{infeasible bin}.
We say that $p$ \emph{uses} a bin~$k$ if $\area(R(k))\neq 0$,
and we say that $p$ is a \emph{feasible configuration} if $\area(R(p_i))\neq
0$ for $i=1,\ldots,n$.  The \emph{cost} of an item $i$ in a configuration
$p$ is given by $c_i = a_i/\area(R(p_i))$ if $\area(R(p_i)) \neq 0$, and is
infinite otherwise.

As in the Elementary Stepwise System dynamics~\cite{Even-darKM07,OrdaRS93}, 
one at a time players selfishly change their choice of bin, in order to 
minimize their cost. That is, a player $i$ in a bin~$k$ has incentive to move 
to a bin~$\ell$ if $i$ fits into bin $\ell$ and $a(R(\ell) \cup \{i\}) > a(R(k))$, 
that is, after the move, its cost decreases.  We note that, if a
configuration is feasible, then an improving step only happens if the
resulting configuration is also feasible.  If a configuration is not
feasible (because $a(R(k))=0$ for some bin~$k$), then an item in an
infeasible bin may migrate to another bin (possibly an empty one) and
decrease its cost. 

The \emph{selfish rectangle packing game} (SRPG) is defined by a set
$\{1,2,\ldots,n\}$ of items (rectangles or players), and a set of unit
bins into which the items are to be packed, and a cost function that
assigns to each item $i$ a cost $c_i$, as defined above. The game
starts with an arbitrary configuration. Each item wants to minimize
its cost, so it can selfishly migrate to decrease its cost. 
The \emph{social cost} is the number of used bins (which is precisely
$c_1+c_2+\cdots+c_n$ if the configuration is feasible).  
When all items to be packed are squares, we refer to the 
\emph{selfish square packing game} (SSPG).

We say that a configuration of a game is a \emph{Nash equilibrium} 
if no player can decrease its cost by moving to (that is, selecting) another bin. 
Since a game starts at an arbitrary configuration, it is reasonable to ask whether 
it will always reach a Nash equilibrium. We say that a game \emph{converges} to 
a Nash equilibrium if the answer to this question is yes. 
The next result also holds for the (more general) selfish 2D bin packing game.

\begin{lemma}\label{lema:convergencia}
  The selfish rectangle packing game converges to a Nash equilibrium.
\end{lemma}

\begin{proof}
  We follow the proof presented by Bil\`o~\cite{Bilo06}, using area instead of
  length. At each moment of a game, consider the list whose elements are the
  numbers $\area(R(k))$ for each used bin $k$, sorted in non-increasing
  order. Note that there is a finite number of different such lists (even
  considering that some bins may be infeasible).  After each migration, the
  new list is always lexicographically greater than the previous. Indeed, when
  a rectangle migrates from a bin $k$ to a bin $\ell$, it is clear that
  $\area(R(\ell))$ after the migration is greater than both $\area(R(k))$ and
  $\area(R(\ell))$ before the migration. So, after a finite number of steps,
  no item can migrate.
\end{proof}

The minimum number of bins needed to pack all the rectangles in a game $G$ is
the \emph{optimal social cost}, denoted simply by $\opt(G)$, or $\opt$ when
$G$ is clear from the context.   Let~${\soc}(p)$ be the number of bins used in a
feasible configuration~$p$. (The acronym sc stands for \emph{social cost}.) 
Let $\caln(G)$ be the set of all Nash equilibria of~$G$. 



The \emph{(asymptotic) price of anarchy} of a class $\calg$ of games is
defined as
\begin{equation}
   \poa := \limsup_{m\rightarrow\infty\;\;} \sup_{G\in\calg,\; \opt(G)=m\; } \max_{p \in \caln(G)} \frac{\soc(p)}{m}.
\end{equation}

The concepts we have defined are also referred to as \emph{pure} price of anarchy and \emph{pure} Nash Equilibrium, but here we omit the term \emph{pure} (that should be understood, if not explicitly stated). 
In this paper, our results on the price of anarchy refer to two classes of games we have already defined: the selfish rectangle packing game (SRPG) and the selfish square packing game (SSPG), both with the cost function we have mentioned (proportional model). We may not write explicitly that the price of anarchy is asymptotic, but we shall always refer to this case.

\medskip

Leung et al.~\cite{LeungTWYC90} showed that the problem of deciding if a set of squares can be packed in a square is NP-complete. So we consider a variant of the game defined above in which a (polynomial-time) packing algorithm $\cala$ is used to decide whether or not a set of items fit together in one bin.  If $\cala$ cannot pack in one bin the rectangles in $R(k)$, then the cost for each~$i$ in $R(k)$ is infinite, otherwise the cost is defined as $a_i/a(R(k))$.  For each packing algorithm $\cala$, we have the corresponding parameterized game, referred to as \emph{selfish rectangle packing game using algorithm~$\cala$}.  In the same way, for these games one can define Nash equilibrium and price of anarchy, denoted by $\poa(\cala)$. Observe that if $\cala$ is an exact packing algorithm, then the corresponding game is precisely SRPG.
If a set of items in a bin $k$ cannot be packed by $\cala$, we say that $k$ is
an {\it infeasible} bin for $\cala$. We make the natural assumption that,
given any item~$i$ and any packing algorithm~$\cala$, the algorithm~$\cala$
succeeds packing $i$ in an empty bin.  A proof analogous to the one presented 
for Lemma~\ref{lema:convergencia} can be used to show that the following holds.

\begin{lemma} 
  For any packing algorithm $\cala$, the selfish rectangle packing game using
  algorithm~$\cala$ converges to a Nash equilibrium.
\end{lemma}

The analogous concept of \emph{strong price of anarchy} will be
defined in Section~\ref{sec-strong}.  We first present results for the
pure price of anarchy, and in Section~\ref{sec-strong} we present
results for the strong price of anarchy.

\section{Bounds for the pure price of anarchy}

This section contains results on the price of anarchy of the selfish rectangle
(and square) packing game. 


\subsection{Selfish rectangle packing}

We start showing that the price of anarchy of the selfish rectangle packing game  is unbounded.

\begin{theorem}
  The price of anarchy of the selfish rectangle packing game is unbounded.
\end{theorem}

\begin{proof}
  Let $p$ be a configuration of the selfish rectangle packing game that uses $k$
  bins, say bins $B_1,\ldots,B_k$. Each bin $B_i$ has exactly two rectangles,
  one with dimension $(1,1/2^i)$ and the other with dimension
  $(1/2^i,1-1/2^i)$. Note that, except for symmetry, the only way to pack the
  two rectangles in one bin is placing the second on top of the first one,
  with total height~$1$ (see Figure~\ref{fiunbounded}(a)). Clearly, a rectangle 
  with dimension $(1,1/2^i)$ cannot migrate to another used bin, as there is no 
  room for it. A rectangle with dimension $(1/2^i,1-1/2^i)$ cannot migrate to a 
  bin $B_j$ with $j<i$, as there is no room for it; and, although it fits into 
  a bin $B_j$ with $j>i$, it has no incentive to move to such $B_j$. Thus, configuration~$p$
  is a Nash equilibrium.  An optimal configuration (shown in Figure~\ref{fiunbounded}(b)) 
  has two bins, one with the items of dimensions $(1,1/2^i)$ and the other with the 
  items with dimensions $(1/2^i,1-1/2^i)$.  So, we have $\soc(p)/2\geq k/2$.
  As the number $k$ of bins can be made as large as we wish, this shows that 
  the price of anarchy of SRPG  is unbounded. 
\end{proof}

\begin{figure}[h!]
  \centerline{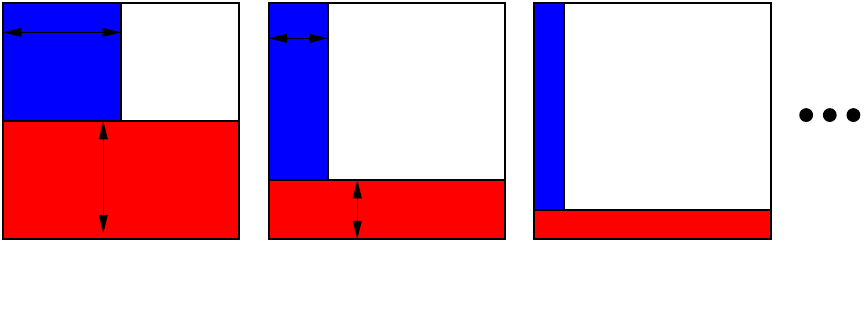\hspace{1.5cm} 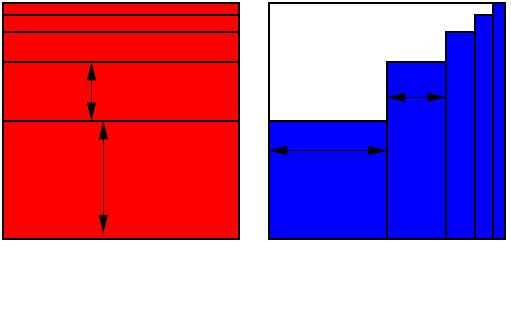~~~~}
  \caption{\label{fiunbounded} (a) A Nash equilibrium using $k$ bins. ~~ (b) Optimal configuration using two bins.}
\end{figure}

Now we show upper bounds for the price of anarchy of the selfish rectangle
packing game using as packing algorithm the well-known \nfdh~\cite{CoffmanGJT80}. 
This algorithm first sorts the set $L$ of rectangles (to be packed) in
decreasing order of height, then packs the rectangles side by side generating
levels. The height of a level is the height of the first rectangle in the
level. When a rectangle cannot be packed in the current level, if it fits in
the same bin on top of the current level, it is packed in a new level above
the previous one. If it does not fit, we say that the \nfdh\ algorithm fails
to pack the rectangles; otherwise, if the whole set $L$ can be packed in this
way in the bin, then we say that \nfdh\ succeeds.

Note that the configuration in Figure~\ref{fiunbounded}(a) is also a Nash
equilibrium of the selfish rectangle packing game using \nfdh.  So $\poa(\nfdh)$ 
is unbounded as well. On the other hand, it is not difficult to show that 
$\poa(\nfdh) \leq 4$ if we allow only rectangles with dimensions at most~$1/2$. 
In fact, in what follows we show an  upper bound for $\poa(\nfdh)$ restricted 
to rectangles with dimensions at most $1/m$, for $m \geq 2$. For that, we use the 
following result, from which Theorem~\ref{thm:bdrect} below can  be easily obtained.

\begin{theorem}[Epstein and Levy \cite{EpsteinL10}]\label{teepsteinlevy} 
  Let $L \cup \{r\}$ be a set of rectangles with dimensions at most $1/m$, for $m \geq 2$. 
  If \nfdh\ packs $L$ in one bin, but cannot pack $L\cup\{r\}$ in one bin, 
  then $\area(L)\geq (\frac{m-1}{m})^2$.
\end{theorem}

\begin{theorem}\label{thm:bdrect}
  For selfish rectangle packing games restricted to rectangles with dimensions 
  at most $1/m$, for $m\geq 2$, we have $\poa(\nfdh)\leq\left(\frac{m}{m-1}\right)^2$.
\end{theorem}


\subsection{Selfish square packing}\label{sec:selfishsquarepacking}

We show in this section a lower bound for the price of anarchy of the
selfish square packing game.  We show a weaker result, and once it is
understood, we mention how a slight improvement can be obtained.

First, we describe an optimal packing $\calo$ that consists of $N$
bins with the same configuration, all completely filled with squares
of side $(1/i) + \eps$, for some different integer values~$i$, and a
very small positive real number $\eps$. (The value of $N$ will be defined
later.) Then, we describe a packing $\calp$ of the same set of squares
that is a Nash equilibrium.

The bins of the optimal packing~$\calo$ have the configuration shown
in Figure~\ref{filowerbound}.  Each bin contains $n_i$ square(s) of
side $s_i$, for $i=1,\ldots, 7$, as defined ahead and shown in
Figure~\ref{filowerbound}; and has the remaining apparently free
space completely filled with small squares, which are like sand.
These small squares have total area $\gamma$ (which is approximately
$0.0488$, considering $\eps$ very small).  These $N$ bins packed in a
completely filled way clearly define an optimal packing.

\hspace{1cm}\begin{minipage}{2.5in}
  \includegraphics[width=2.0in]{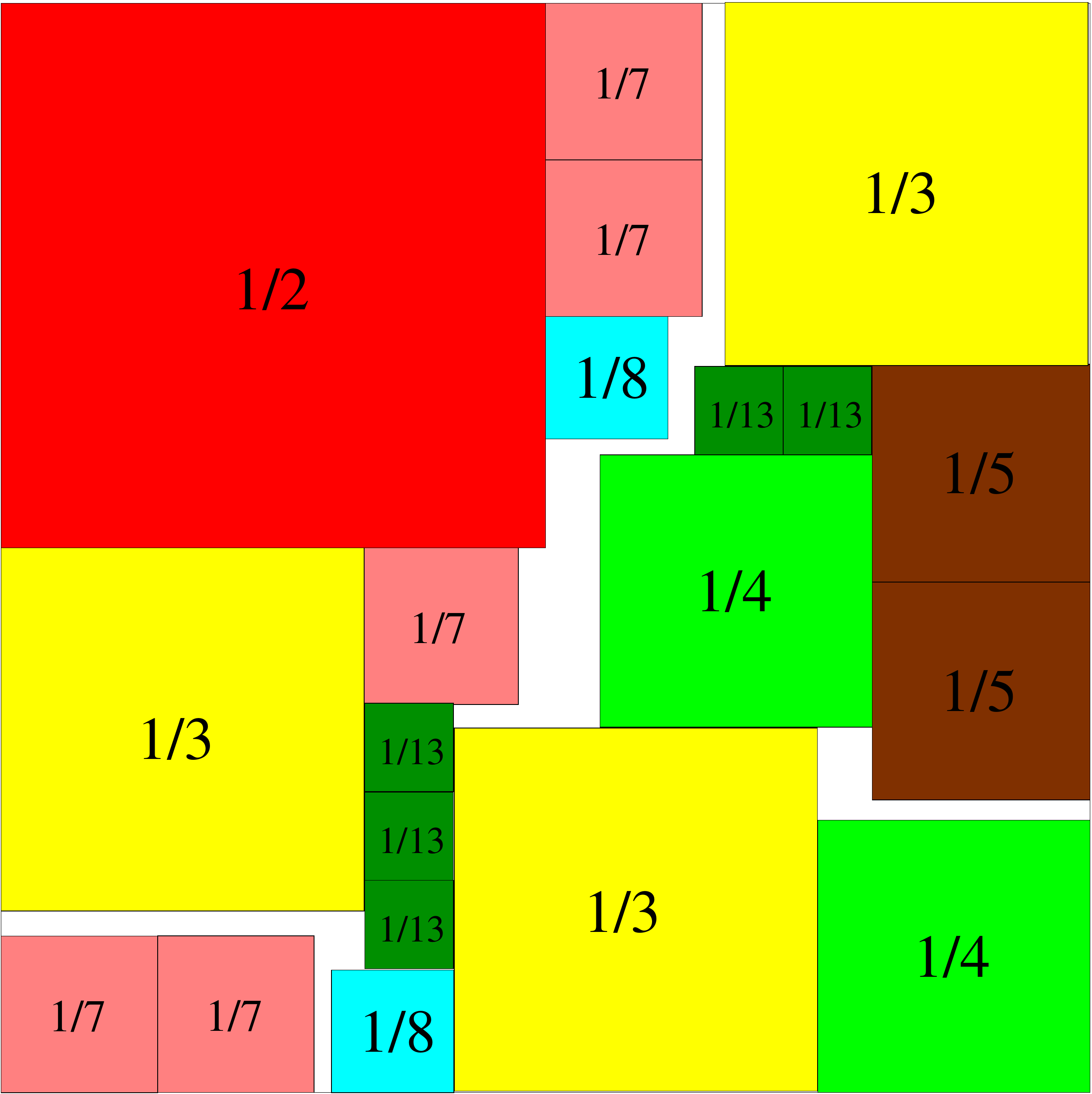} 
  \captionof{figure}{\label{filowerbound}\footnotesize{An optimal packing.\\
                                (Fraction $1/i$ stands for $1/i + \eps$.)}}
\end{minipage}\hspace*{0cm}\begin{minipage}{11cm}
  \hspace*{1.3cm}$n_1:=1$; side $s_1:=1/2 + \eps$; \\
  \hspace*{1.3cm}$n_2:=3$; side $s_2:=1/3 + \eps$; \\
  \hspace*{1.3cm}$n_3:=2$;  side $s_3:=1/4 + \eps$; \\
  \hspace*{1.3cm}$n_4:=2$; side $s_4:=1/5 + \eps$; \\
  \hspace*{1.3cm}$n_5:=5$; $s_5:=1/7 + \eps$; \\
  \hspace*{1.3cm}$n_6:=2$; side $s_6:=1/8 + \eps$; \\
  \hspace*{1.3cm}$n_7:=5$; side $s_7:=1/13 + \eps$; and \\
  \hspace*{1.3cm}very small squares of the same side, \\
  \hspace*{1.3cm}with total area $\gamma:=1-\sum_{i=1}^7n_is_i^2$.\\
\end{minipage}\medskip

Now, let us describe the packing $\calp$. 
It consists of homogeneous bins. For $0 < \ell \leq 1$, we say that an
occupied bin~$B$ is \emph{$\ell$-homogeneous} if all squares packed
in~$B$ have side~$\ell$, and it has the maximal number of such squares.

Note that, for an integer $i>1$, the wasted space is large for a 
$(1/i + \eps)$-homogeneous bin for small $i$ and $\eps$.
Moreover, no square packed in a $(1/j + \eps)$-homogeneous 
bin can migrate to another $(1/i + \eps)$-homogeneous bin.  
Indeed, if $j \leq i$, the square does not fit in the other bin. 
If $j > i$, then, in a $(1/i + \eps)$-homogeneous bin, there are 
$(i-1)^2$ squares and the occupied area is $(i-1)^2(1/i + \eps)^2$. 
Moreover, as $j \geq i+1$, one can check that 
$(1/j+\eps)^2 + (i-1)^2(1/i + \eps)^2 < (j-1)^2(1/j + \eps)^2,$ for $\eps$ small enough.  
So, a configuration consisting of homogeneous bins is a Nash equilibrium.

Now, we define how many homogeneous bin of each type there are, so
that the set of squares packed in the $N$ bins is precisely the set of
squares packed in these homogeneous bins. These numbers are:
\begin {quote}
  $m_1:=N$ bins, each one with precisely one square of side $1/2 + \eps$;\\
  $m_2:=3N/4$ bins, each one with precisely 4 squares of side $1/3 + \eps$;\\
  $m_3:=2N/9$ bins, each one  with precisely 9 squares of side $1/4 + \eps$;\\
  $m_4:=2N/16$ bins, each one with precisely 16 squares of side $1/5 + \eps$;\\
  $m_5:=5N/36$ bins, each one  with precisely 36 squares of side $1/7 + \eps$; \\
  $m_6:=2N/49$ bins, each one  with precisely 49 squares of side $1/8 + \eps$;\\
  $m_7:=5N/144$ bins, each one  with precisely 144 squares of side $1/13 + \eps$; and\\
  $m_8:=\gamma N$ bins, each one completely filled with the very small squares.
\end{quote}
Take $N$ such that each $m_i$ is an integer.  
This configuration is a Nash equilibrium (because of what was argued before), 
and uses $\sum_{i=1}^8m_i > 2.3604\,N$ bins.  
As $\opt=N$, we conclude that $\poa > 2.3604$.

We obtained the configuration in Figure~\ref{filowerbound} in the
following way. We packed first a square of side $1/2+\eps$, then we
packed the largest possible number of squares of side $1/3+\eps$, then
continued packing the largest possible number of squares of side
$1/4+\eps$, and so on. We stopped the process with squares of side
$1/13+\eps$, as the gain from continuing became very small. At this point,
we filled the remaining area with equal size tiny squares (like sand)
as that helps to increase a little more the obtained ratio, and also
because it makes it obvious that the configuration described is an optimal packing.

This result was obtained before we learned that similar ingredients
were used by Epstein and van Stee~\cite{EpsteinS07} to provide a lower
bound for the online bounded space hypercube packing problem.  
In fact, for square packing, the configuration obtained by these authors
gives a better bound for our game problem. While we stopped with
squares of side~$1/13+\eps$, they proceed up to squares of
side~$1/43+\eps$, in the following way (continuing from $n_7$):
$n_8:=2$ squares of side $s_8:=1/14 + \eps$; $n_9:=1$ square of side
$s_9:=1/18 + \eps$; $n_{10}:=2$ squares of side $s_{10}:=1/21 + \eps$;
$n_{11}:=2$ square of side $s_{11}:=1/25+ \eps$; $n_{12}:=24$ squares
of side $s_{12}:=1/43 + \eps$; and likewise considering ``sand'' to
fill the remaining free space.
Considering the configuration just described, and the results obtained by 
Epstein and van Stee~\cite{EpsteinS07}, using the same arguments presented above, 
calculations lead to an improvement on the third decimal digit, as we state below. 
For completeness, and for simplicity, we presented the proof above that gives the 
lower bound $2.3604$.

\begin{theorem}\label{thm236}
  The price of anarchy of the selfish square packing game is at least $2.3634$.
\end{theorem}



\medskip 

We now turn to upper bounds for the price of anarchy of the selfish square
packing game. For that, let us consider a result that is similar to
Theorem~\ref{teepsteinlevy}, proved by Meir and Moser~\cite[Thm.~1]{MeirM68}
for the special case in which all items are squares.

\begin{theorem}[Meir and Moser \cite{MeirM68}]\label{NewRef1} 
  Every set $S$ of squares whose largest square has side $\ell$ can be packed 
  by \nfdh\ in a rectangle with dimension $(b,h)$ if $\ell \leq \min\{b, h\}$ 
  and $\area(S) \leq \ell^2 + (b-\ell)(h-\ell)$.
\end{theorem}

\begin{corollary}\label{temeirmoser} 
  Let $S$ be a set of squares, each with side at most $1/m$ for $m \geq 2$.
  If ${a(S) \leq \frac{1}{m^2}+(1-\frac{1}{m})^2}$, then \nfdh\ can pack $S$ in
  a unit square.
\end{corollary}

Using this result, we can prove that $\poa(\nfdh)\leq 4$. For that, it
suffices to note that all bins, except possibly one, have an occupied
area of at least~$1/4$. This area occupation is valid for bins containing 
one item of side greater than $1/2$, and also for the other bins containing 
only squares with side at most $1/2$ (by Theorem~\ref{teepsteinlevy}). 
In what follows, we will prove a better upper bound for $\poa(\nfdh)$.

\smallskip

The next lemma is the main part of the proof of a result stated without 
a proof in the preliminary version of this paper~\cite[Lemma 3.3]{FFMW11}.  
Its proof, as well as that of the following theorem, uses the following
terminology. 

A square is {\it big} if it has side larger than $1/2$; {\it medium} if it 
has side larger than $1/3$ and at most $1/2$; and {\it small} if it has side 
at most~$1/3$. 

\begin{lemma}\label{lequatronono}
  Let $y$ be a small square and $X$ be a set of small or medium squares. 
  If $\nfdh$ cannot pack $X\cup\{y\}$ in one bin, then $\area(X)\geq 4/9$.
\end{lemma}

\begin{proof}
  For simplicity, for a square~$x$, we use $x$ to refer to its side.  
  Suppose by contradiction that $\area(X) < 4/9$. This in particular 
  means that $X$ contains at most three medium squares. 

  Let us first argue that $X$ contains at least one medium square. If every
  square in $X$ is small, then we can apply Corollary~\ref{temeirmoser} to $X
  \cup \{y\}$, with $m=3$. Indeed, as $\area(X \cup \{y\}) < 4/9 + 1/9 = 5/9 =
  1/m^2 + (1-1/m)^2$, \nfdh\ packs all squares from $X \cup \{y\}$ in a bin, a
  contradiction. So~$X$ contains at least one medium square.

  Now let us prove that $X$ should contain exactly three medium squares.
  Suppose by contradiction that $X$ contains one or two medium squares. Let
  $X^+= X \cup \{y\}$. By assumption, \nfdh\ fails to pack~$X^+$. Let us analyze
  why this happens.

  Observe that all squares have side at most $1/2$, thus \nfdh\ fails to pack a
  square only after the second level. Let $Z$ be the set of squares in $X^+$
  that \nfdh\ packs in the first level.  Let $x$ be a largest square in $Z$, and
  $r$ be a largest square in $X^+\setminus Z$. Note that $r \leq 1/3$, because
  $X^+$ contains at most two medium squares, which are in $Z$.
  
  Let $R$ be the rectangle of dimension $(1,1-x)$ (corresponding to the region
  of a bin above a first level of height $x$). As \nfdh\ fails to pack $X^+$ in
  a bin, it also fails to pack $X^+\setminus Z$ in $R$. Hence, by
  Theorem~\ref{NewRef1}, $\area(X^+\setminus Z) > r^2 + (1-r)(1-x-r)$.
  Thus,
  \begin{eqnarray*}
    \area(X) = \area(X^+) - y^2  > r^2 + (1-r)(1-x-r) + \area(Z) - y^2 
             \geq  r^2 + (1-r)(1-x-r) + x^2.
  \end{eqnarray*}

  The last inequality follows from the fact that $\area(Z)-y^2 \geq x^2$.
  This holds because there are at least two squares in $Z$ and the second
  largest square in $Z$ is at least as large as $y$.  Hence,
  \begin{eqnarray*}
    \area(X) & > & x^2 + 2r^2 - x - 2r + xr + 1 \\
             & = & (x - 1/3)^2 + 2(1/3 - r)^2 - x(1/3 - r)  - (2/3) r + 2/3 \\
             & \geq & 2/3- 1/2(1/3 - r) - (2/3) r  \quad \mbox{(because $x \leq 1/2$)} \\
             & = & 1/2 - r/6 \ \geq \ 4/9 \quad \mbox{because $r \leq 1/3$.}
  \end{eqnarray*}
  But this is a contradiction, because $a(X) < 4/9$.

  So $X$ contains three medium squares. As before, let $X^+= X \cup \{y\}$. By
  assumption, \nfdh\ fails to pack~$X^+$. Let us analyze why this happens. 
  Again, \nfdh\ fails to pack a square in $X^+$ only after the second
  level.  Also, as $X^+$ contains three squares in $M$, \nfdh\ packs exactly two
  squares in the first level (two of the squares in $X^+ \cap M$). Let us argue
  that \nfdh\ packs at least three squares in the second level.

  Let $x_1,x_2,\ldots$ be the squares in $X$ sorted so that $x_1 \geq x_2 \geq
  \cdots$ and let $x_i=0$ if $i > |X|$. If $X$ had at most three squares, \nfdh\
  would suceed in packing~$X^+$. So $X$ has at least four squares, and $|X^+|
  \geq 5$. To prove that there are at least three squares in the second level,
  it is enough to show that $x_3 + x_4 \leq 2/3$. Indeed, if $x_3 + x_4 \leq 2/3$,
  as $y \leq 1/3$ and $x_5 \leq 1/3$, clearly $x_3 + x_4 + \max\{y,x_5\} \leq 1$. 
  This means that \nfdh\ packs at least three squares in the second level. 
  (Hence $|X| \geq 5$, otherwise \nfdh\ would not fail to pack $X^+$.)

  In fact, we will prove a stronger statement that will help us also at another
  point of the proof. We will prove (by contradiction) that $x_1 + x_3 + x_4 \leq 1$. 
  This implies that $x_3 + x_4 \leq 1 - x_1 < 2/3$.  Suppose that $x_1+x_3+x_4 > 1$, 
  that is, $x_4 > 1-x_1-x_3$. Then
  $$\area(X)  >  x_1^2 + 2x_3^2 + (1-x_1-x_3)^2 
              =  2(x_1-1/3)^2 + 3(x_3-1/3)^2 + 4/9 + 2x_1x_3 - (2/3)x_1.$$

  But $2x_1x_3 - (2/3)x_1 \geq 0$ because $x_3 > 1/3$.  
  So, we can conclude that $\area(X) > 4/9$, a contradiction. 
  Therefore $x_1+x_3+x_4 \leq 1$, as we wished to prove.

  Let $Z$ be the set of squares in $X^+$ that \nfdh\ packs in the first two levels. 
  As \nfdh\ fails to pack $X^+$, the set $X^+ \setminus Z$ is non-empty. 
  Let $r$ be a largest square in $X^+ \setminus Z$.  Note that $r \leq x_4$, 
  because $x_4$ is certainly packed in the second level. Thus $r \leq 1-x_1-x_3$ 
  because $x_1+x_3+x_4 \leq 1$.

  The heights of the first and second levels are $x_1$ and $x_3$ respectively. 
  So let $R$ be the rectangle of dimension $(1,1-x_1-x_3)$, corresponding to the 
  region of the bin above the first two levels.  As \nfdh\ fails to pack $X^+$, 
  we know that it also fails to pack $X^+ \setminus Z$ in~$R$. Hence, using 
  Theorem~\ref{NewRef1}, we derive that $\area(X^+\setminus Z) > r^2 + (1-r)(1-x_1-x_3-r)$, 
  and
  \begin{eqnarray}
    \area(X) & = & \area(X^+) - y^2 \nonumber \\
             & > & r^2 + (1-r)(1-x_1-x_3-r) + \area(Z) - y^2 \nonumber \\ 
             & \geq & x_1^2 + 2x_3^2 + 2r^2 + (1-r)(1-x_1-x_3-r). \label{eq00}
  \end{eqnarray}
  Inequality \eqref{eq00} holds because $\area(Z)-y^2$ is at least the sum of the
  areas of the four largest squares in $X$, which is at least $x_1^2+2x_3^2+r^2$.
  Now, as $r \leq x_1$, from the inequalities above, we get that 
  \begin{eqnarray}
    \area(X) & > & x_1^2 + 2x_3^2 + 2r^2 + (1-x_1)(1-x_1-x_3-r) \nonumber \\ 
             & = & x_1^2 + (1-x_1)^2 + 2x_3^2 + 2r^2 - (1-x_1)(x_3+r) \nonumber \\ 
             & \geq & 1/2 + 2x_3^2 + 2r^2 - (1-x_3)(x_3+r) \label{eq01} \\ 
             & = & 3x_3^2 + 2r^2 - x_3 - r + x_3r + 1/2 \nonumber \\ 
             & = & 3(x_3-1/6)^2 + 2(r-1/6)^2 + x_3r - r/3 +13/36 \nonumber \\
             & \geq & 3(1/6)^2 + r/3 - r/3 + {13}/{36} \ = 4/9. \label{eq02}
  \end{eqnarray}
  Note that inequality~\eqref{eq01} holds because $x_1^2 + (1-x_1)^2 \geq 1/2$ 
  for every~$x_1$, and \eqref{eq02} holds for every $x_3 \geq 1/3$. 
  The equalities are plain mathematical manipulation.  But $\area(X) > 4/9$ 
  contradicts the initial hypothesis, and we conclude the proof that, 
  if $\nfdh$ cannot pack $X\cup\{y\}$ in one bin, then $\area(X)\geq 4/9$.
\end{proof}

Note that the bound $4/9$ in Lemma~\ref{lequatronono} is best possible. 
To see this, consider $X$ consisting of four squares of side $1/3+\eps$, 
and $y$ having side $1/3$.

\begin{theorem}\label{lesquaresarea}
  For the selfish square packing game with all squares having side at most
  $1/2$, in any Nash equilibrium, the area used in each bin is at least $4/9$,
  except for at most two bins.
\end{theorem}

\begin{proof}
  Consider a Nash equilibrium of such a game. Let us call \emph{bad} 
  a used bin whose used area is less than~$4/9$. We will prove that the
  number of bad bins is at most two.

  There is at most one bad bin that contains only medium squares. This is
  true because any bad bin can contain at most three medium squares. As
  \nfdh\ suceeds in packing any four medium squares in a bin, there cannot
  exist two bad bins containing only medium squares (because if they existed,
  then any square from the less used bin could migrate to the most used one).

  Let us prove that there is at most one bad bin with at least one small
  square. If we prove this, the number of bad bins is at most two (one with
  only medium squares and the other containing at least one small square).

  Suppose there are two bad bins, $B_X$ and $B_Y$, with at least one small
  square each. Let $X=R(B_X)$ and $Y=R(B_Y)$. Suppose $\area(X) \geq \area(Y)$,
  and let $y$ be a small square in~$Y$. By Lemma~\ref{lequatronono}, if $y$
  cannot migrate to $B_X$, then $\area(X)\geq 4/9$, which is a contradiction.
\end{proof}

Observe that the theorem cannot be improved to state that all but one bin have
occupied area of at least $4/9$. Indeed, consider one bin with one square of
side $1/2$ and another bin with three squares of side $1/3+\eps$ and two of
side $1/6$. Both such bins are bad, and together define  a configuration
which is a Nash equilibrium.

The following lemma will be used in the proof of the next two theorems.
\begin{lemma}[Lemma 4.4~\cite{KohayakawaMRW2004}]\label{lem:gamadelta}
  For every real numbers $a$, $b$, $\gamma$, $\delta$ with $a > 0$ and $0 < \gamma < \delta < 1$, 
  $$ \frac{a+b}{\max\{a, \gamma a + \delta b\}} \ \le \ 1 + \frac{1-\gamma}\delta. $$
\end{lemma}

Now we are ready to prove the promised improved upper bound for $\poa(\nfdh)$.


\begin{theorem}\label{thm315}
  For the selfish square packing game, $\poa(\nfdh)\leq 2.6875$.
\end{theorem}

\begin{proof}
  Consider a configuration $p$ that is a Nash equilibrium for the selfish square packing game
  using~\nfdh. Let $\calb$ be the set of bins in this configuration containing a big square,
  and let $\calc$ be the set of the remaining used bins.
  Bins in $\calb$ have an occupied area of at least~$1/4$. By
  Theorem~\ref{lesquaresarea}, at least $|\calc|-2$ bins in $\calc$
  have an occupied area of at least $4/9$. Set $N_B= |\calb|$ and
  $N_C = |\calc|-2$. Thus,
  $\opt \geq \area(\calb) + \area(\calc) \geq
  \frac{1}{4}N_B+\frac{4}{9}N_C$.
  On the other hand, $\opt \geq N_B$, as each big square has to be
  packed in a different bin.  Thus,
  $\opt \geq \max\{N_B; \frac{1}{4}N_B+\frac{4}{9}N_C\}$.

  Putting together these results, we have that
  $$\frac{sc(p)}{\opt}  =  \frac{N_B+N_C+2}{\opt}
           \leq \frac{N_B+N_C}{\max\{N_B\, ; \, \frac{1}{4}N_B+\frac{4}{9}N_C\}}+\frac{2}{\opt}
           \leq  \frac{43}{16}+\frac{2}{\opt} = 2.6875+\frac{2}{\opt},$$
where the last inequality follows from Lemma~\ref{lem:gamadelta}. 
\end{proof}
 
Observe that Theorem~\ref{thm315} also holds for any packing algorithm
for which the result stated in Lemma~\ref{lequatronono} holds and that
succeeds in packing any four medium squares into a bin. In particular,
Theorem~\ref{thm315} holds for an exact packing algorithm. 

\medskip

We note that all results presented in this section for games using
NFDH also hold for the SRPG.  This is because all upper bounds were 
proved using area occupation arguments, which hold also for the SRPG. 
It should be noted, however, that results on the upper bound 
for the price of anarchy of a game parameterized by an algorithm may not 
hold for the non-parameterized game. In the particular case of NFDH, the obtained results
imply the following result.

\begin{corollary}\label{corlowerupper}
  For the selfish square packing game,  $2.3634\leq \poa \leq 2.6875$.
\end{corollary}

\smallskip



\section{Bounds for the strong price of anarchy}\label{sec-strong}

In this section, we consider the same cost function mentioned in
the previous section, but we study \emph{strong} Nash equilibria. 
We have the same setting, and the game is analogous, but now coalition 
is allowed, that is, a group of items (not necessarily from a common bin) 
may migrate at once if it leads to a cost that is better for all items in the group. 
It is also allowed that a group of items migrates to a new bin 
(that is, the number of used bins may increase).


We call \emph{strong Nash equilibrium} a configuration of the
selfish rectangle packing game in which no group of items has
incentive to move to decrease the individual cost of its items.
Let $\calsn(G)$ be the set of all strong Nash equilibria of~$G$.
For strong Nash equilibria, analogous to the case of Nash equilibria, the following 
measures are considered. 
The \emph{asymptotic strong price of anarchy} of a class $\calg$ of games is
\begin{equation}
 \spoa:= \limsup_{m\rightarrow\infty\;\;} \sup_{G\in\calg,\; \opt(G)=m\; } \max_{p \in \caln(G)} \frac{\soc(p)}{m}.
\end{equation}


When the game is parameterized by a packing algorithm $\cala$, we denote the strong price of anarchy of the corresponding game by $\spoa(\cala)$.

\subsection{The strong price of anarchy is bounded}

We first note that, as opposed to the pure price of anarchy, the strong 
price of anarchy is bounded.  This is easy to see because one can
guarantee that, in any strong Nash equilibrium, each bin has an occupied 
area of at least~$1/4$, except perhaps for a constant number of bins. 
In fact, this is valid also when the algorithm $\nfdh$ is used.
 
We use the following result, proved by Harren and van Stee~\cite{HarrenV08}, 
that is a generalization of a theorem of Meir and Moser.

\begin{lemma}\label{lemeirm2} 
  Let $T$ be a set of rectangles, and $w_{\max}$ (resp.\ $h_{\max}$) be the maximum width 
  (resp.\ height) of a rectangle in $T$.  If $T$ is packed into a rectangle $R=(a,b)$
  by $\nfdh$, then either a total area of at least $(a-w_{\max})(b-h_{\max})$ is packed 
  or the algorithm runs out of items, i.e., all items are packed.
\end{lemma}

\subsection{Lower and upper bounds for the strong price of anarchy}

We first present a lower bound for the strong price of anarchy of the 
selfish square packing game.  Before that, we note that the 
configuration described in Section~\ref{sec:selfishsquarepacking}, 
consisting of homogeneous bins and used to prove a lower bound on the price of anarchy for this game,
is not a strong Nash equilibrium.  Indeed, for instance, as long as 
$\eps$ is small enough, any group of seven squares of side $1/5+\eps$ 
has incentive to migrate to a $(1/4+\eps)$-homogeneous bin.

\begin{theorem}
  The strong price of anarchy of the selfish square packing game is at least $2.076$.
\end{theorem}

\begin{proof}
  For each integer $k \geq 3$, we obtain a lower bound $L_k$ for the 
  strong price of anarchy of the selfish square packing game in which 
  $k$ specific sizes of squares are considered.  The idea is 
  similar to the one used in Section~\ref{sec:selfishsquarepacking}, 
  but requires a more judicious choice of possible square sides.
  For $i=1,\ldots,k-1$, we consider squares of sides $\sigma_i=(1+\eps)/2^i$, 
  where $\eps$ is a very small positive real number. We also consider 
  squares of side $\sigma_k$, a positive real number that divides $1$ 
  and $\sigma_{k-1}$ (and therefore divides any positive value of the 
  form $1-\sum_{i=1}^{k-1} a_i\sigma_i$ for integer values of $a_i$).
  

  First, we describe an optimal packing $\calo$ that consists of $N$ bins 
  with the same configuration, all completely filled with squares of sides 
  $\sigma_1,\ldots,\sigma_k$.  (The value of $N$ will be defined later.) 
  Then, we describe a packing $\calp$ of the same set of squares, 
  consisting of only homogeneous bins, that is a strong Nash equilibrium. 

  The bins of the optimal packing~$\calo$ have the configuration shown
  in Figure~\ref{stronglb1}.  Each bin contains $n_i$ square(s) of
  side $\sigma_i$, for $i=1,\ldots,k-1$, as defined below, and has the
  remaining apparently free space completely filled with small squares 
  of side $\sigma_k$, which are like sand, as long as $\eps$ is small enough.  
  These $N$ bins packed in a completely filled way clearly define an optimal packing.

\begin{figure}[!ht]
  \begin{center}
    \scalebox{0.2}{\includegraphics{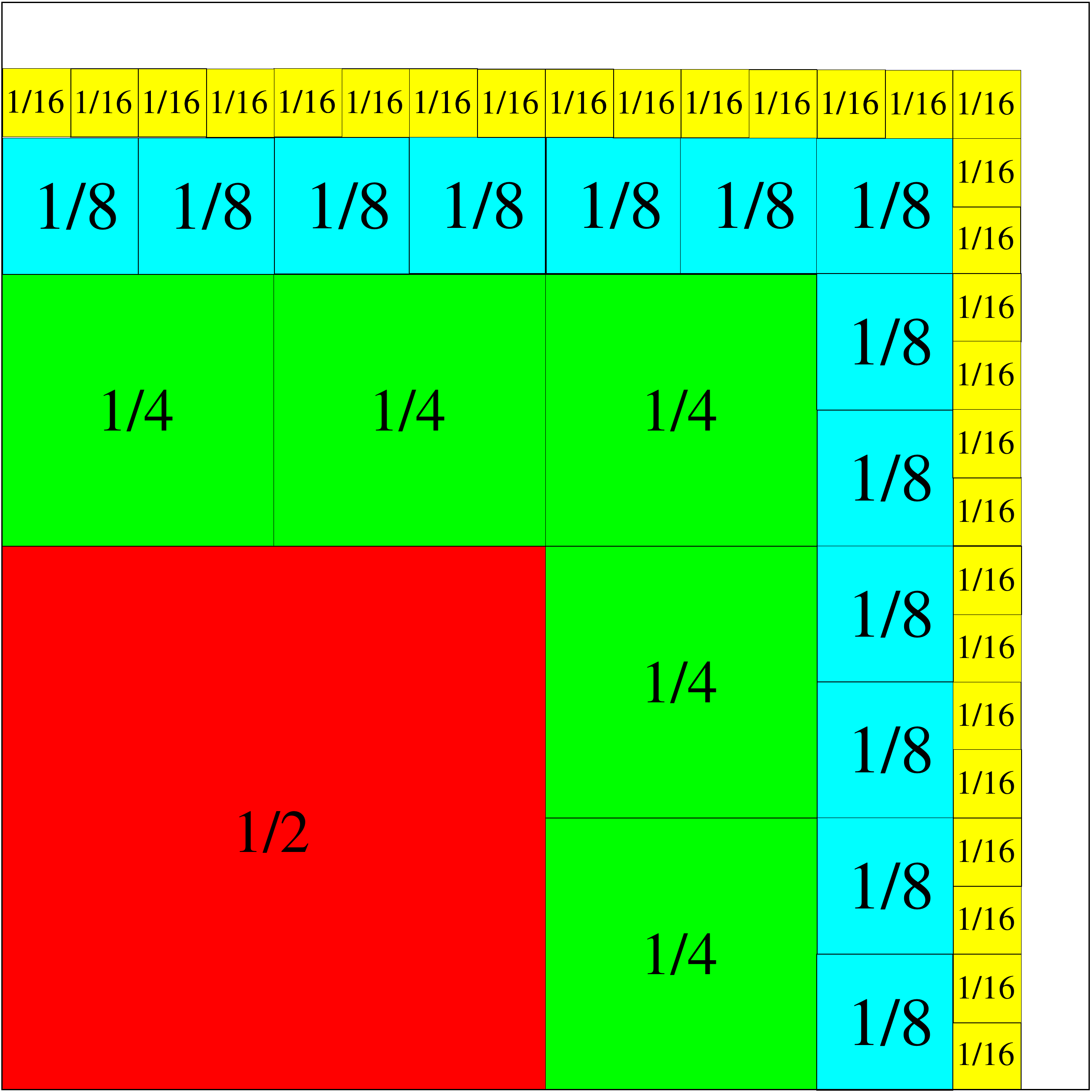}}
    \caption{\label{stronglb1}Configuration of a bin in the optimal packing $\calo$. 
      Each fraction $1/2^i$ represents a square of side $\sigma_i = (1+\eps)/2^i$.}
  \end{center}
\end{figure}

  The maximum number of squares of side $(1+\eps)/2^i$ that pack in a bin is $(2^i-1)^2$.  
  So, by packing always the maximum number of squares of side $\sigma_i$ before starting
  to pack squares of side $\sigma_{i+1}$, we will be able to pack 
  $n_i = (2^i-1)^2 - 4(2^{i-1}-1)^2 = 2^{i+1}-3$ squares of side $\sigma_i$, for $i = 1,\ldots,k-1$. 
  For instance, as shown in Figure~\ref{stronglb1}, $n_1 = 1$, $n_2 = 2^3-3 = 5$, 
  $n_3 = 2^4 - 3 = 13$, and $n_4 = 2^5 - 3 = 29$.  As for $n_k$, it is set so that 
  each bin in $\calo$ is completely filled, that is, 
  $n_k = (1-(2^{k-1}-1)^2\sigma_{k-1}^2)/\sigma_k^2$.

  Now, let us describe the packing $\calp$.  
  It consists of a number of $k$ different types of homogeneous bins,
  one for each side $\sigma_i$.  Since we can pack $(2^i-1)^2$ squares 
  of side $\sigma_i$ in a bin, for ${i=1,\ldots,k-1}$, the number of 
  $\sigma_i$-homogeneous bins in $\calp$ is $m_i = n_i N / (2^i-1)^2$, 
  for $i=1,\ldots,k-1$.  The maximum number of squares of side $\sigma_k$ 
  that pack in a bin is $1/\sigma_k^2$, as $\sigma_k$ divides 1, so 
  ${m_k = n_k N \sigma_k^2 = (1-(2^{k-1}-1)^2\sigma_{k-1}^2)N}$.  
  The value of $N$ is such that every $m_i$ is an integer (for that, 
  it suffices that $N$ be a multiple of $(2^i-1)^2$ for $i=1,\ldots,k-1$).

  We claim that $\calp$ is a strong Nash equilibrium. First, note that
  items of type $\sigma_k$ do not have incentive to migrate, as, in $\calp$, 
  all such items are in bins completely filled.  Now let $S$ be a group of 
  squares whose sides are in $\{\sigma_1,\ldots,\sigma_{k-1}\}$ and suppose 
  that all squares in $S$ have incentive to migrate to a bin $B$.  Let us 
  derive a contradiction by proving that there is a square in $S$ that has 
  no incentive to migrate to $B$.  Let $i$ be such that the side of the 
  smallest square in $S$ is $\sigma_i$.  Clearly, $B$ is a 
  $\sigma_j$-homogeneous bin with $j < i$, otherwise no square in $S$
  would fit in a bin with all the squares in $B$.  But then $\sigma_i$ 
  divides $\sigma_j$ and also the side of each square in $S$.  Thus
  $B \cup S$ corresponds to a set of squares all of them of side $\sigma_i$,
  and, if $B \cup S$ can be packed in a bin, such bin would have occupied
  area at most equal to a $\sigma_i$-homogeneous bin in $\calp$.  Therefore 
  any square in $S$ of side $\sigma_i$ has no incentive to migrate to $B$.
  We conclude that $\calp$ is indeed a strong Nash equilibrium.

  For each $k$, a lower bound $L_k$ is obtained by considering 
  the previously described optimal packing $\calo$, consisting of $N$ bins, 
  and the strong Nash equilibrium $\calp$, consisting of $m_1+\cdots+m_k$
  bins. More precisely, for $\eps = 1/(2^{k-1}-1)$, we have that $n_k = m_k = 0$ 
  and
  $$L_k \ = \ \frac{|\calp|}{|\calo|} \ = \ \frac{m_1+\cdots+m_{k-1}}{N} 
       \ = \ \sum_{i=1}^{k-1} \frac{n_i}{(2^i-1)^2} 
       \ = \ \sum_{i=1}^{k-1} \frac{2^{i+1}-3}{(2^i-1)^2}.$$
  Thus, for instance, as $L_{20} > 2.076$, the theorem follows. 
  Moreover, $\lim_{k \rightarrow \infty} L_k < 2.0761$. 
\end{proof}

\bigskip

Now, we show an upper bound for the strong price of anarchy of the
selfish square packing game, parameterized by the $\nfdh$ algorithm.
As we have previously defined, we recall that squares with sides in
the intervals $(0,1/3]$, $(1/3,1/2]$, and $(1/2,1]$ are called \textit
{small}, \textit{medium}, and \textit{big}, respectively.

First, we show a minimum area occupation of bins in a configuration
that is a strong Nash equilibrium, when all squares have side at
most~$1/k$, for $k\geq 2$.

\begin{lemma}\label{lek}
  Let $\cals$ be a strong Nash equilibrium in which all squares have 
  side at most $1/k$, for an integer $k\geq 2$. Then, each bin in $\cals$, 
  except for at most $k^2$ of them, has an occupied area of at least $(k/(k+1))^2$.
\end{lemma}

\begin{proof}
  Let $\cals'$ be the set of bins in $\cals$ with at least one square of side in 
  $(1/(k+1),1/k]$ and $\cals''$ be the remaining bins in~$\cals$.  First, let us argue 
  that there is at most one bin in $\cals''$ with area occupation less than $(k/(k+1))^2$. 
  Theorem~\ref{NewRef1} guarantees that we can pack in one bin any set of squares whose 
  largest side is at most $1/(k+1)$ and the total area is at most $1/(k+1)^2+(1-1/(k+1))^2$.  
  Thus, except for at most one, every bin in $\cals''$ has an area occupation of at least 
  $(1-1/(k+1))^2$, that is, at least $(k/(k+1))^2$.  Second, note that we can pack in one 
  bin $k^2$ squares that have side in~$(1/(k+1),1/k]$, leading to a total area of at least 
  $(k/(k+1))^2$.  Therefore, we can have at most $k^2-1$ bins in $\cals'$ with occupied area 
  less than $(k/(k+1))^2$, otherwise there would be $k^2$ such squares that could migrate to 
  a new bin, with better area occupation.
\end{proof}

\begin{theorem}
  For the selfish square packing game, $\spoa(\nfdh)\leq 2.3605$.  
\end{theorem}

\begin{proof}
  Let $p$ be a configuration that is a strong Nash equilibrium (SNE), and 
  let $L$ denote the set of squares packed in the bins of this configuration.  
  Denote by~$\calb$ the set of bins in the configuration~$p$ that contain 
  a big square (and possibly other squares), $\calm$ the set of bins containing 
  at least one medium square and no big square, and $\cals$ the set of bins 
  containing only small squares.

  Since every strong Nash equilibrium is also a Nash equilibrium, 
  all results that we have previously proved for Nash equilibria are 
  also valid here.  In particular, by Theorem~\ref{lesquaresarea}, we know 
  that all bins in $\calm$, except possibly two of them, have an occupied 
  area of at least~$4/9$.  By Lemma~\ref{lek}, all bins in $\cals$, except
  possibly nine of them, have an occupied area of at least $9/16$. Since all 
  bins in~$\calb$ have an occupied area larger than $1/4$, the following 
  inequalities hold:
  \begin{eqnarray*}
    \opt(L) & \geq & \area(L) \ = \ \area(\calb) + \area(\calm) + \area(\cals) \\ 
            & \geq & \frac{1}{4}|\calb| + \frac{4}{9}(|\calm|-2) + \frac{9}{16}(|\cals|-9).
  \end{eqnarray*}
  
  Next, we show that some of the above fractions (on least occupied area) can be 
  improved, giving a better lower bound for $\opt(L)$. For that, we
  analyze two cases. In what follows, we refer to a value $\valor$ which will be specified
  later. For the moment, assume $0.62<\valor<0.63$.

  \medskip
  
  \noindent \textbf{Case 1.} There is at least one bin $B$ in~$\calb$ containing 
            solely one big square of side at most $\valor$ and no other squares.

  \medskip
  
  \begin{quote}
    We shall prove that, in this case, all bins in $\calm$, except 
    for at most three of them, have an occupied area of at least~$4(1-\valor)^2$.
    To prove this claim, denote by $\calm'$ the set of bins in~$\calm$
    containing at least one medium square of side at most $1-\valor$, 
    and let $\calm'' = \calm \setminus \calm'$.
    
    First we show that all bins in $\calm'$, except for at most two 
    of them, have an occupied area of at least~$7/12$.  Indeed, suppose 
    that there are three bins in $\calm'$ whose occupied area is less 
    than~$7/12$.  Then, three squares of side at most $1-\valor$, each 
    coming from one of these three bins, can migrate to the bin $B$ 
    in $\calb$ and, after this migration, the occupied area of~$B$ 
    becomes larger than $1/4 + 3 (1/9) = 1/4 + 1/3 = 7/12$, 
    contradicting the fact that $p$ is an SNE.


    Now we show that all bins in $\calm''$, except for at most three 
    of them, have an occupied area larger than~$4(1-\valor)^2$.
    Indeed, suppose there are four bins in~$\calm''$ whose occupied 
    area is less than or equal to~$4(1-\valor)^2$.
    In this case, four medium squares with side larger than $1-\valor$,
    each one coming from one of these four bins, can migrate to a new (empty) 
    bin, say~$B'$, because, after the migration, $a(B') > 4(1-\valor)^2$,
    a contradiction, as $p$ is an~SNE.

    We will choose $\valor$ so that $7/12 > 9/16 \geq 4(1-\valor)^2$. 
    Thus we conclude that all bins in $\calm$, except for at most five of them, 
    have an occupied area of at least~$4(1-\valor)^2$.
    Hence, 
    \begin{eqnarray*}
      \opt(L) &\geq& \area(L) \ = \ \area(\calb) +  \area(\calm) + \area(\cals) \\
              &\geq& \frac{1}{4}|\calb| + 4(1-\valor)^2(|\calm|-5) + \frac{9}{16}(|\cals|-9) \\
              &\geq&  \frac{1}{4}|\calb| + 4(1-\valor)^2(|\calm| + |\cals|-14).
    \end{eqnarray*}
    As $\opt(L) \geq |\calb|$, we have that
    $\opt(L) \geq \max \{|\calb|, \frac{1}{4} |\calb| + 4(1-\valor)^2 (|\calm|+|\cals| - 14)\}$. 
    So,
    \begin{eqnarray*}
      \frac{sc(p)}{\opt} = \frac{|\calb| + |\calm| + |\cals|}{\opt} 
      & \leq & \frac{|\calb|+|\calm|+|\cals|-14}{
               \max\{|\calb|, \, \frac{1}{4}|\calb| + 4(1-\valor)^2(|\calm|+|\cals|-14)\}} + \frac{14}{\opt} \\
      & \leq & 1 + \frac3{16(1-\valor)^2} + \frac{14}{\opt},
    \end{eqnarray*}
    by Lemma~\ref{lem:gamadelta}.
  \end{quote}
  
  \medskip
  
  \noindent \textbf{Case 2.} There is no bin in $\calb$ containing solely
                             one big square of side at most $\valor$.
  
  \medskip
  \begin{quote}
    In this case, we shall prove that the bins in $\calb$ have an occupied area of 
    at least $\valor^2$, improving on the trivial occupation bound of~$1/4$.
    
    \medskip
    

    \noindent Claim A: \textit{All bins in $\calb$ containing at least one medium 
    square have an occupied area of at least $\valor^2$, except for at most two of them.}

    \medskip
    
    \noindent \emph{Proof of the claim.}
    Suppose there are three bins $B_1$, $B_2$, and $B_3$ in $\calb$, where 
    $B_i$ contains a big square $b_i$ and a medium square $m_i$ for $i=1,\ldots,3$, 
    with occupied area smaller than $7/12$.  Assume that $b_1\leq b_2\leq b_3$.
    In this case, the set of squares $\{b_1,m_1,m_2,m_3\}$ can migrate
    to a new bin, that will have an occupied area of at least
    $1/4+3(1/9)=7/12>\valor^2$, a contradiction, as $p$ is an SNE.\medskip

    \noindent Claim B: \textit{All bins in $\calb$ containing no medium square 
       have an occupied area of at least~$\valor^2$, except for at most 18 of them.}

    \medskip
    
    \noindent \emph{Proof of the claim.}
    For each bin $B \in \calb$, denote by $b_B$ the big square in $B$.
    Every bin~$b$ in~$\calb$ with $b_B \geq \valor$ has occupied area at least $\valor^2$.  
    So let $\calb'$ be the set of bins $B$ in~$\calb$ with no medium square and 
    $b_B < \valor$.  If $\calb'$ has at most 17 bins, the claim follows.  
    Thus suppose there are at least 18 bins in $\calb'$.  By the hypothesis 
    of Case~2, there is at least one small square in each bin in $\calb'$.  
    Let $B'$ be a bin in~$\calb'$ with $a(B')$ minimum and let~$B$ be a bin in 
    $\calb'$ distinct from $B'$.  Let~$\calr_B=(1,1-b_B)$ be the rectangular region 
    above the square~$b_B$ in $B$.  We must have that $a(B \setminus \{b_B\}) \geq 
    (1-1/3)(1-b_B-1/3)$, otherwise each small item from $B'$ would have incentive 
    to migrate to $B$, as it fits next to $b_B$ and \nfdh\ would be able to pack 
    $B \setminus \{b_B\}$ in the rectangular region~$\calr_B$ by Lemma~\ref{lemeirm2}. 
    As $(1-1/3)(1-b_B-1/3) \geq (1-1/3)(1-\valor-1/3)$, we claim that the occupied 
    area of each bin in $\calb'$, except for possibly $B'$ and at most $17$ other bins 
    in~$\calb'$, is at least $\valor^2$.  Otherwise, the small squares in these 17 bins 
    either pack together in a bin, occupying an area of at least 
    $17(1-1/3)(1-\valor-1/3) \geq 17 \cdot 0.02444 = 0.41548 > \valor^2$, or a proper 
    subset of them packs in a bin, occupying an area of at least $4/9 > \valor^2$, by 
    Lemma~\ref{lequatronono}.  So either all of these small squares or a proper subset 
    of them would have incentive to migrate to a new bin, contradicting the 
    fact that $p$ is a SNE.

    From Claims A and B, we have that all bins in $\calb$ have area least $\valor^2$, except 
    for at most 19 bins. 
    Proceeding now with the analysis of Case~2, we have that 
    \begin{eqnarray*}
      \opt(L) & \geq & \area(L) \ = \ \area(\calb) + \area(\calm) + \area(\cals) \\
              & \geq & \valor^2(|\calb|-19) + \frac{4}{9}(|\calm|-2) + \frac{9}{16}(|\cals|-9) \\
              & \geq & \valor^2(|\calb|-19) + \frac{4}{9}(|\calm| + |\cals|-11).
    \end{eqnarray*}
    As $\opt(L) \geq |\calb|$, we have 
    $\opt(L) \geq \max \{|\calb|-19, \valor^2(|\calb|-19) + \frac{4}{9} (|\calm|+|\cals|-11)\}$. 
    Thus, by Lemma~\ref{lem:gamadelta},
    \begin{eqnarray*}
      \frac{sc(p)}{\opt} = \frac{|\calb| + |\calm| + |\cals|}{\opt} 
      & \leq & \frac{|\calb|+|\calm|+|\cals|-30}
               {\max\{|\calb|-19, \valor^2(|\calb|-19) + \frac{4}{9} (|\calm|+|\cals|-11)\}} 
               + \frac{30}{\opt} \\
      & \leq & 1 + \frac{9(1-\valor^2)}4 + \frac{30}{\opt}.
    \end{eqnarray*}
  \end{quote}

  Making the constant part of the bound of the two cases equal, we obtain the equation
  ${(1-\valor^2)(1-\valor)^2 = 1/12}$, that has only one solution in the interval $[0,1]$, 
  namely, ${x \approx 0.62876}$. We note that the restriction
  $9/16\geq 4(1-\valor)^2$ mentioned in Case 1 is satisfied for this value of $\valor$. 

  Thus, using this value of $\valor$, we conclude that $\spoa(\nfdh) \leq 2.3605$.
\end{proof}

\bigskip

As in the previous section, on the pure price of anarchy, the results 
we have presented for games using NFDH also hold for the SRPG, because 
all upper bounds were proved using area occupation arguments. Hence, the
following result on the strong price of anarchy also holds. 

\begin{corollary}
  For the selfish square packing game, $2.076 \leq \spoa \leq 2.3605$. 
\end{corollary}

\section{Concluding remarks}

The parameterization of the selfish packing game by a specific packing
algorithm has shown to be interesting in various ways. First,
depending on the algorithm and the analysis of the (pure or strong)
price of anarchy of the corresponding game, the result that one
obtains carries over to the generic game. As we have shown, in the
case of square packing, the upper bound for $\poa(\nfdh)$ and for 
$\spoa(\nfdh)$ have led us to results for~$\poa$ and $\spoa$.  
This indicates that it would be interesting to consider other 
packing algorithms and study the behavior of the corresponding
games. Possibly, better upper bounds for~$\poa$ and $\spoa$ can be
obtained this way.  Another way of viewing the results for such
parameterized games is to consider that they also may be seen as
a coordination mechanism to be used (to control selfish decisions of
the players), specially if one can show that the price of anarchy of
the corresponding game has an acceptable upper bound. In this case, 
it is desirable to have easy to be implemented polynomial-time
algorithms, and this is the case of the algorithms we have considered.


The results we have shown for the $2$-dimensional case, with the very
natural cost function (proportional to the occupied area), are the
first ones in the literature. It would be interesting to further
improve the bounds we have shown here, to tighten the gap.  Another
challenge would be to study the natural generalization of this game to
the $d$-dimensional case, for $d \geq 3$. 

Also, there are several variants of the problems considered here, 
such as no repacking versions, items of other forms, like disks, 
rectangles with rotations allowed, etc. Each such variant requires 
different strategies and are challenging problems.


\bibliography{games}

\begin{thebibliography}{10}

\bibitem{Bilo06}
V.~Bil\`o.
\newblock On the packing of selfish items.
\newblock In {\em Proc. of the 20th Internacional Parallel and Distributed
  Processing Symposium (IPDPS)}, pages 9--18. IEEE, 2006.

\bibitem{CoffmanGJT80}
E.~G. {Coffman, Jr.}, M.~R. Garey, D.~S. Johnson, and R.~E. Tarjan.
\newblock Performance bounds for level oriented two-dimensional packing
  algorithms.
\newblock {\em SIAM Journal on Computing}, 9:808--826, 1980.

\bibitem{Epstein13}
L.~Epstein.
\newblock Bin packing games with selfish items.
\newblock In {\em Proc.\ of the 38th International Symposium on Mathematical
  Foundations of Computer Science (MFCS)}, pages 8--21, 2013.

\bibitem{EpsteinK11}
L.~Epstein and E.~Kleiman.
\newblock Selfish bin packing.
\newblock {\em Algorithmica}, 60(2):368--394, 2011.

\bibitem{EpsteinL10}
L.~Epstein and M.~Levy.
\newblock Dynamic multi-dimensional bin packing.
\newblock {\em Journal of Discrete Algorithms}, 8(4):356--372, 2010.

\bibitem{EpsteinS07}
L.~Epstein and R.~{van Stee}.
\newblock Bounds for online space hypercube packing.
\newblock {\em Discrete Optimization}, 4(2):185--197, 2007.

\bibitem{Even-darKM07}
E.~Even-Dar, A.~Kesselman, and Y.~Mansour.
\newblock Convergence time to {N}ash equilibrium in load balancing.
\newblock {\em ACM Transactions on Algorithms}, 3(3):Article 32, 2007.

\bibitem{FFMW11}
C.G. Fernandes, C.E. Ferreira, F.K. Miyazawa, and Y.~Wakabayashi.
\newblock Selfish square packing.
\newblock In {\em Electronic Notes in Discrete Mathematics}, volume~37, pages
  369--374. Elsevier, 2011.
\newblock \emph{Proc.\ of the VI Latin-American Algorithms, Graphs and
  Optimization Symposium}.

\bibitem{HarrenV08}
R.~Harren and R.~{van Stee}.
\newblock Packing rectangles into {2OPT} bins using rotations.
\newblock In {\em Proc. of the 11th Scandinavian Workshop on Algorithm Theory
  (SWAT)}, pages 306--318, 2008.

\bibitem{KohayakawaMRW2004}
Y.~Kohayakawa, F.K. Miyazawa, P.~Raghavan, and Y.~Wakabayashi.
\newblock Multidimensional cube packing.
\newblock {\em Algorithmica}, 40(3):173--187, 2004.

\bibitem{KoutsoupiasP99}
E.~Koutsoupias and C.~H. Papadimitriou.
\newblock Worst-case equilibria.
\newblock In {\em Proc. of the 16th Annual Symposium on Theoretical Aspects of
  Computer Science (STACS)}, pages 404--413, 1999.

\bibitem{KoutsoupiasP09}
E.~Koutsoupias and C.~H. Papadimitriou.
\newblock Worst-case equilibria.
\newblock {\em Computer Science Review}, 3(2):65--69, 2009.

\bibitem{LeungTWYC90}
J.~Y-T. Leung, T.~W. Tam, C.~S. Wong, G.~H. Young, and F.~Y.~L. Chin.
\newblock Packing squares into a square.
\newblock {\em J. of Parallel and Distributed Computing}, 10:271--275, 1990.

\bibitem{MaDHTYZ13}
R.~Ma, G.~D\'osa, X.~Han, H.-F. Ting, D.~Ye, and Y.~Zhang.
\newblock A note on a selfish bin packing problem.
\newblock {\em Journal of Global Optimization}, 56(4):1457--1462, 2013.

\bibitem{MeirM68}
A.~Meir and L.~Moser.
\newblock On packing of squares and cubes.
\newblock {\em J. Combinatorial Theory Ser. A}, 5:116--127, 1968.

\bibitem{OrdaRS93}
A.~Orda, R.~Rom, and N.~Shimkin.
\newblock Competitive routing in multiuser communication networks.
\newblock {\em IEEE/ACM Transactions on Networking}, 1(5):510--521, 1993.

\bibitem{Papadimitriou01}
C.H. Papadimitriou.
\newblock Algorithms, games, and the internet.
\newblock In {\em Proc. of the 33rd Annual ACM Symposium on the Theory of
  Computing (STOC)}, pages 749--753, 2001.

\bibitem{YuZ08}
G.~Yu and G.~Zhang.
\newblock Bin packing of selfish items.
\newblock In {\em Proc. of the 4th International Workshop on Internet and
  Network Economics}, pages 446--453, 2008.

\end{thebibliography}

\end{document}